\documentclass[11pt]{article}
\usepackage{amssymb, latexsym, mathtools, amsthm, amsmath}
\usepackage{amsfonts}
\usepackage{stmaryrd}
\usepackage{authblk}
\renewenvironment{abstract}
 { \normalsize
  \list{}{\setlength{\leftmargin}{.0cm}%
    \setlength{\rightmargin}{\leftmargin}}%
  \item {\bf \abstractname.}\relax}
 {\endlist}
\usepackage{xcolor,colortbl}
\usepackage{enumerate}
\usepackage{pdfsync}
\usepackage[numbers]{natbib}
\usepackage{parskip}
\theoremstyle{plain}
\newtheorem{thm}{Theorem}[section]
\newtheorem{prop}[thm]{Proposition}

\newtheorem{lem}[thm]{Lemma}

\newtheorem{coro}[thm]{Corollary}

\theoremstyle{definition}

\newtheorem{defi}[thm]{Definition}

\newcommand{\Nat}{\mathbb{N}}

\newcommand{\restr}{\upharpoonright}  

\newcommand{\un}{\uparrow} 
\newcommand{\de}{\downarrow} 

\DeclarePairedDelimiter{\tuple}{\langle}{\rangle}

\newcommand{\sqbrad}[2]{\{\hspace{0.03cm}{#1} : {#2}\hspace{0.03cm}\}}

\DeclarePairedDelimiter{\dbra}{\llbracket}{\rrbracket}

\newcommand{\bigo}[1]{\textsf{O}\hspace{0.01cm}\big({#1}\big)}

 \usepackage{tcolorbox}
 \newcommand{\mybox}[1]{\begin{tcolorbox}[width=\textwidth,colback={gray!10!white},boxrule=1pt, parbox=false] {#1}\end{tcolorbox}}

\newcommand{\abs}[1]{|{#1}|}

\newcommand{\parb}[1]{\big({#1}\big)}

\newcommand{\ml}{Martin-L\"{o}f }

\newcommand{\pz}{$\Pi^0_1$\ }

\newcommand{\ie}{i.e.\ }
\newcommand{\ce}{c.e.\ }

\newcommand{\pf}{prefix-free }

\newcommand{\twome}{2^{\omega}}

\newcommand{\zj}{\emptyset'}
\newcommand{\twomel}{2^{<\omega}}

\newcommand{\wedga}{\ \wedge\ \ }

\newcommand{\leqT}{\leq_T}
\newcommand{\geqT}{\geq_T}

\newcommand{\impl}{\implies}

\newcommand{\hthree}{\hspace{0.3cm}}

\newcommand{\sz}{$\Sigma^0_1$\ }
\newcommand{\szn}{$\Sigma^0_1$}

\newtheorem*{thme}{Theorem}
\usepackage[section]{algorithm}
\usepackage{algorithmic}
\usepackage{algorithm,caption}

\newcommand{\subseteqm}{\subseteq\hspace{-0.04cm}}
\newcommand{\inv}{^{-1}}
\newcommand{\zerome}{0^{\omega}}

\title{Computable one-way functions on the reals\thanks{
Authors are in alphabetical order. We thank L. Levin for his guidance and suggestions.
Supported by Beijing Natural Science Foundation (IS24013).}}
\author{George Barmpalias} \author{Xiaoyan Zhang}  
\affil{State Key Lab of Computer Science\\ \vspace{0.1cm} 
Institute of Software, Chinese Academy of Sciences\\ \vspace{0.1cm}
University of  Chinese Academy of Sciences}

\begin{document}
\maketitle
\begin{abstract}
A major open problem in computational complexity is the existence of a one-way function, 
namely a function from strings to strings which  is computationally easy to compute  but hard to invert.
Levin (2023) formulated the notion of one-way functions
from reals (infinite bit-sequences) to reals  in terms of computability, and 
asked whether partial computable one-way functions exist. 
We give a strong positive answer using the hardness of the halting problem and
exhibiting a total computable one-way function.
\end{abstract}

\section{Introduction}\label{Ugcvsz7Hdh}
A function is {\em one-way} if it is computationally easy to compute but hard to invert, even probabilistically.
Many aspects of computer science  such as computational complexity, pseudorandom generators and digital signatures
rely on one-way functions from {\em strings} (finite binary sequences) to strings \cite{levin24ZF, Levin2003}. Although their existence is not known,
much research effort has  focused on their implications and relations with other fundamental  problems in computation
\cite{Impagliazzo1989,Impagliazzo1990, Hstad1999}. Recent research has established strong connections between one-way functions and
Kolmogorov complexity \cite{Segev2023, Liu2020,Liu2023, Hirahara2023}.

Levin \cite[\S 2.2]{levin24ZF} asked for the existence of a one-way function 
from {\em reals} (infinite binary sequences) to reals, which he defined in terms of
 probabilistic computability  with respect 
 to the uniform measure $\mu$.\footnote{This was also discussed in \cite{LevinEmailDec23}.} 
 These are partial computable functions $f$ 
 which  preserve algorithmic randomness and 
\[
\textrm{the probability that $M$ inverts  $f$, namely $f(M(y))=y$, is zero}
\]
for each probabilistic Turing machine $M$. We give a positive answer:

\begin{thme}
There exists a total computable one-way surjection $f$.
\end{thme}

The hardness of inverting $f$ is based on the hardness of the halting problem $\zj$. As a result, 
every randomized continuous inversion of $f$ computes $\zj$.

\citet{Gacs24oneway} independently  constructed  $f$
which is `one-way'  with respect to  the domain instead of the range of $f$:
for each partial computable $g$, 
\begin{equation}\label{QnjppKpZCa}
\mu\parb{\sqbrad{(x,r)}{f(g(f(x),r))=f(x)}}=0
\end{equation}
and $f$ is partial computable with domain of positive uniform measure $\mu$.

Although \eqref{QnjppKpZCa} appears to be a closer analogue to the
one-way functions in computational complexity, it does not preserve randomness, which is required
by Levin \cite{levin24ZF}. Effective partial maps that meet \eqref{QnjppKpZCa} often have a null image.
Levin's definition implies \eqref{QnjppKpZCa} while the converse fails \cite[\S 3.2]{owrand24}.

In \S\ref{4GY8Jx1fU} we justify Levin's definition as the correct analogue of one-way functions on discrete domains,
over weaker formalizations. We exhibit a one-way real function using a framework 
that can be adapted toward meeting additional conditions, and explore their properties.

We do not know if there is an injective one-way real function. 
In \S\ref{3p5CxCWq7F}
we explore the extent to which injectivity is compatible with one-way real functions and
end with a summary and some problems in \S\ref{UWucNoldLo}.

\section{Preliminaries}\label{WlGeuf4gYD}
Let $\Nat$ be the set of natural numbers, represented by $n,m, i,j,t,s$. Let 
\begin{itemize}
\item $2^\omega$ be the set of  reals, represented by variables $x,y,z,v,w$
\item $2^{<\omega}$ the set of  strings, represented by variables $\sigma,\tau, \rho$.
\end{itemize}
We index the bits $x(i)$ of $x$ starting from $i=0$. Let 
\begin{itemize}
\item  $x\restr_n$ be the $n$-bit prefix $x(0)x(1)\cdots x(n-1)$ of $x$ 
\item $\preceq$, $\prec$ be the prefix and strict prefix relation between strings
\item $\sigma\mid \tau$  denotes that $\sigma\not\preceq\tau\ \wedge\ \tau\not\preceq\sigma$
\end{itemize}
and $\succeq, \succ$ denote the suffix relations which, along with the prefix relations can also apply between a string and a real.
Given  $x, y$ let
\[
\textrm{ $x\oplus y:=z$ where $z(2n)=x(n)$ and $z(2n+1)=y(n)$}
\]
and similarly for strings of the same length. 

\subsection{Computability and randomness}
The \textit{Cantor space} is $2^\omega$  with the topology generated by the 
basic open sets
\[
\dbra{\sigma}:=\sqbrad{z\in\twome}{\sigma\prec z}
\hspace{0.3cm}\textrm{for $\sigma\in 2^{<\omega}$}
\]
which we  call {\em cylinders}.
Let $\mu$ be the  \textit{uniform measure} on $\twome$, determined by $\mu(\dbra{\sigma})=2^{-|\sigma|}$.  
Probability in $\twome\times\twome$ is reduced
to $\twome$ via the measure-preserving  $(x,y)\mapsto x\oplus y$. Classes  of $\mu$-measure 0
are called {\em null}.

A {\em tree} is a downward $\preceq$-closed  $T\subseteq\twomel$ and 
\begin{itemize}
\item  $T$ is {\em pruned} if every $\sigma\in T$ has an extension in $T$ 
\item  $T$ is {\em perfect} if every $\sigma\in T$ has two extensions $\tau|\rho$ in $T$.
\item $x$ is a {\em path} through $T$ if all of its prefixes belong to $T$.
\end{itemize}
Let $[T]$ be the class of all paths through $T$. 

We use the standard notion of relative computability  in terms of Turing machines with oracles from $\twome$.
Turing reducibility $x\leqT z$ means that $x$ is computable from $z$ (is {\em $z$-computable}) and is a preorder  calibrating  $\twome$ according to computational power
in the Turing degrees. 

Effectively open sets,  also known as $\Sigma^0_1$ classes, are subsets of $\twome$ of the form
\[
\bigcup_i\dbra{\sigma_i}\hspace{0.3cm}\textrm{where $(\sigma_i)$ is computable.}
\]
Effectively closed sets or \pz classes are the complements of \sz classes. 
Every \pz class is the set of paths through some computable tree, and vice versa. 

If $(\sigma_{n,i})$ is computable then 
\begin{itemize}
\item classes $V_n=\bigcup_i\dbra{\sigma_{n,i}}$ are called uniformly $\Sigma^0_1$ 
\item $\bigcap_n V_n$ is called a $\Pi^0_2$ class and its complement is a $\Sigma^0_2$ class.
\end{itemize}
Equivalently,  $V\in \Pi^0_2$ iff there is a computable predicate $P$ such that 
\[
x\in V\iff \forall n\ \exists s\ P(x\upharpoonright_n,x\upharpoonright_s).
\]
Relativization to an oracle $r$ defines $\Sigma^0_1(r), \Pi^0_2(r)$ classes and so on.

A \textit{Martin-L\"of test} is a uniformly $\Sigma^0_1$ sequence $(V_n)$ with $\mu(V_n)\leq 2^{-n}$. 
\begin{defi}
Given $k\in\Nat$, a  real $x$ is 
\begin{itemize}
\item \textit{random} if $x\notin\bigcap_n V_n$ for any Martin-L\"of test $(V_n)$
\item \textit{weakly random} if it is in every $\Sigma^0_1$ set of measure $1$. 
\item \textit{weakly $k$-random} if it is in every $\Sigma^0_k$ set of measure $1$. 
\end{itemize}
Relativization to  oracle $r$ defines $r$-random, weakly $r$-random  and so on.
\end{defi}
By the countable additivity of the uniform measure $\mu$:
\begin{equation}\label{VbNXoA2gzg}
\textrm{$x$ is weakly $r$-random iff it is not in any null $\Sigma^0_2(r)$ class}
\end{equation}
and similarly for weakly $k$-random reals $x$.

\subsection{Computable analysis}
A Turing machine with a one-way infinite output tape and access to an oracle from $\twome$ may eventually print
an infinite binary sequence, hence defining a  partial map from reals to reals. This  standard notion of computability of real functions 
\cite{PR89} is almost as old as computability itself  \cite{Gherardi11T} and implies that computable real functions are continuous.
Let
\begin{itemize}
\item $f:\subseteqm \twome\to\twome$ denote that $f$ is a partial map from $\twome$ to $\twome$
\item  $f(x)\de, f(x)\un$ denote that $f(x)$ is defined or undefined
\item $f(x;n):=f(x)(n)$ denote bit $n$ of $f(x)$
\end{itemize}
and $f(x;n)\de, f(x;n)\un$ denote that $f(x;n)$ is defined or undefined. 

If $f$ is partial computable, the {\em oracle-use} of $f(x;n)$  
 is the prefix of $x$ that has been read by the underlying Turing machine at the time where
when  $f(x;n)$ is printed on position $n$ of the one-way output tape.
\begin{defi}
We say that $f:\subseteqm \twome\to\twome$ is {\em  random-preserving} if $f(x)$ is random for each random $x$
in the domain of $f$.
\end{defi}
The range and the inverse images of $f$ are denoted by
\[
f(\twome):=\sqbrad{y}{\exists x,\ f(x)=y}
\hspace{0.3cm}\textrm{and}\hspace{0.3cm}
f\inv(y):=\sqbrad{x}{f(x)=y}.
\]
If $f$ is total and continuous  the inverse images
\[
f^{-1}(\dbra{\sigma}):=\sqbrad{x}{\sigma\prec f(x)}
\]
are closed and effectively closed if $f$ is also computable. Continuous functions $f: \subseteq \twome\to \twome$ 
can be defined via a {\em representation}: a $\subseteq$-monotone map between cylinders. 
Formally, let $\lim_{\tau\prec x} \hat f(\tau)=z$ denote that
\[
\forall \tau\prec x,\ \hat f(\tau)\prec z
\hspace{0.5cm}\textrm{and}\hspace{0.5cm}
 \lim_{\tau\prec x} \abs{\hat f(\tau)}=\infty.
\]
Given a continuous $f:\subseteqm\twome\to\twome$ we say that $g:\subseteqm\twome\to\twome$:  
\begin{itemize}
\item {\em  inverts $f$ on} $y$  if  $g(y)\de$ and $f(g(y))\de=y$
\item is an {\em inversion of $f$} if it is continuous and inverts $f$ on all $y\in f(\twome)$.
\end{itemize}
We say that $\hat f: \twomel\to \twomel$  is a {\em representation of $f$}  if
\begin{itemize}
\item $\sigma\preceq\tau \impl \hat f(\sigma)\preceq \hat f(\tau)$
\item $f(x)\de\iff  \lim_{\tau\prec x} \hat f(\tau)= f(x)\iff \lim_{\tau\prec x} \abs{\hat  f(\tau)}=\infty$.
\end{itemize}
Every continuous $f: \subseteq \twome\to \twome$ has a representation.
Every partial computable $f$  has a computable representation.
Under an effective coding of the graphs of representations 
into $\twome$ we identify them with their codes.
\begin{defi}
We say that a continuous $f:\subseteqm \twome\to\twome$  {\em computes $z$} and  denote it by $z\leqT f$ 
if every representation  of $f$ computes $z$.
\end{defi}
Note that by \cite{Miller2004} there is no {\em canonical way} to assign a 
representation to each continuous $f$: there are continuous $f$ such that
for each representation $z$ of $f$ there is another one of lesser Turing degree.

\section{One-way functions}\label{4GY8Jx1fU}
Toward foundational research that is relevant to mathematical practice, Levin proposed a new 
direction \cite{levin24ZF}.\footnote{This includes 15 drafts developed  from 2022 to the end of 2024.} 
He used invertibility of real functions to express the axioms of choice and powerset, 
and emphasized the significance of its computational hardness in the context of his proposal.

 He formally extends the notion of one-way functions to continuous domains and he asks about their existence. 
One may formalize this notion in several ways which are not equivalent, especially in the case of partial computable maps. 
In \S\ref{2JVBnaASm} we give the formal definition from \cite{levin24ZF} and explain why it is the correct
 extension of the discrete notion over weaker alternatives.

In \S\ref{W6lu6GL2ui} we construct a total computable one-way surjection. 
We do not know whether there is a one-way injection on the reals. 
In \S\ref{IDdXsdgbYm} we establish properties
of one-way real functions relating to this question, setting the basis for the comprehensive analysis in
\S\ref{3p5CxCWq7F}.

\subsection{Levin's one-way functions on the reals}\label{2JVBnaASm}
In computational complexity a function from strings to strings is 
{\em one-way} if it is  easy to compute and hard to invert, even probabilistically.

By Levin \cite{levin24ZF} a one-way real function $f$  must 
\begin{enumerate}[\hthree(a)]
\item be partial computable with domain of positive measure
\item have no effective probabilistic inversion
\item satisfy $\mu(f\inv(\dbra{\tau}))=\bigo{\mu(\dbra{\tau})}$.
\end{enumerate}
Conditions (a), (b) clearly correspond to conditions in the discrete one-way functions.
Probabilistic inversions of $f$ are facilitated by
$g:\subseteq\twome\times\twome\to\twome$
where the secondary argument represents access to a randomness source.
However probability in (b) may refer to the domain or the range of $f$.

 Levin \cite{levin24ZF} chose the latter, interpreting ``$g$ inverts $f$'' as the event
 that 
\begin{equation}\label{VLu8JMYa7K}
\textrm{$f(g(y,r))\de =y$ on  randomly chosen $y\in f(\twome)$, $r\in\twome$.}
\end{equation}
Formally, the {\em probability that $g$ inverts $f$}  is the measure of
\[
L_{f}(g):=\sqbrad{y\oplus r}{g(y,r)\de\wedga f(g(y,r))\de=y}
\]
and $g$  is a {\em randomized inversion} of $f$ if  $\mu(L_{f}(g))>0$.
\begin{defi}[Levin \cite{levin24ZF}]\label{vCsTo1DfVQ}
We say that $f:\subseteqm\twome\to\twome$ is {\em one-way} if it 
\begin{itemize}
\item is partial computable, random-preserving with positive domain
\item does not have any partial computable randomized inversion.
\end{itemize}
\end{defi}
The alternative is condition \eqref{QnjppKpZCa} of G{\'a}cs \cite{Gacs24oneway} which, as discussed in \S\ref{Ugcvsz7Hdh},
is strictly weaker if we assume that $f$ is random-preserving.

The mild measure-preservation condition (c) 
is equivalent to randomness-preservation
and implies  $\mu(f(\twome))>0$  for partial computable $f$ (see \cite[\S 3.2]{owrand24}).
To see why (c) is essential  consider the weaker conditions:
\begin{enumerate}[(i)]
\item $\mu(f(E))>0$ for each subset $E$ of the domain of $f$ with $\mu(E)>0$
\item the domain of $f$ contains $R$ with $\mu(R)>0$ and $\mu(f(R))>0$
\item there is a  $\zj$-random $x$ such that $f(x)$ is random 
\end{enumerate}
and note that (c)$\to$(i)$\to$(ii)$\to$(iii) while none of the converses holds.
It is not hard to show (see \cite[\S 3.2]{owrand24}) that each (i), (ii), (iii) is  equivalent to (c)
up to effective restrictions of $f$. So essentially (c) only asks that $f$ maps at least one sufficiently random real to a random real.

Weaker versions based on combinations of (i), (ii), (iii), $\mu(f(\twome))>0$ 
are less robust  as  they do not form a discernible hierarchy \cite{owrand24}.

\subsection{Construction of a total one-way surjection}\label{W6lu6GL2ui}
We will map input $x$ to a permutation $y$ of selected bits of $x$.
Let 
\begin{itemize}
\item $\tuple{\cdot, \cdot}: \Nat\times\Nat\to\Nat$ be a computable bijection with $\tuple{n,s}\geq s$
\item $(\zj_s)$ be an effective enumeration of $\zj$ without repetitions. 
\end{itemize}
We give a simplified construction which we extend in \S\ref{W6lu6GL2ui} for the full result.
\begin{prop}\label{9w41OY8XJy}
There exists a total computable  $f:\twome\to\twome$ such that any inversion of $f$ computes $\zj$.
\end{prop}\begin{proof}
We show that the following $f$  has the required properties:
\[
f(x;\tuple{n,s}):=
\begin{cases}
x(n) & n\in \zj_s-\zj_{s-1} \\
\ 0 & \rm{otherwise.}
\end{cases}
\]
Clearly $f$ is total computable.
Assuming that 
$g:\subseteqm \twome\to\twome$ is an inversion of $f$, we show how to decide if 
$n\in \zj$ using the computation of  $g(0^{\omega};n)$.

Since $y:=0^{\omega}$ is in the range of $f$, the real $x:=g(y)$ is defined.
Let $u_n$ be the oracle-use in the computation of $g(y;n)$, that is, the minimum $u$ such that the computation of $g(y;n)$ reads only $y\upharpoonright_u$. 
It remains to show that 
\begin{enumerate}[\hthree(a)]
\item if $x(n)=1$ then $n\not\in \zj$
\item if $x(n)=0$ then $n\in \zj\iff n\in \zj_{u_n}$. 
\end{enumerate}
If $x(n)=1$ then  $\forall s\ f(x;\tuple{n,s})=y(\tuple{n,s})=0$ so 
$\forall s\  n\not\in \zj_s$ and $n\not\in \zj$.

For (b) assume that $x(n)=0$ and for a contradiction let $n\in \zj_s-\zj_{s-1}$ for some $s>u_n$.
Then there is $z\in\twome$ with $z=f(0^n 1 0^{\omega})=0^{\tuple{n,s}} 1 0^{\omega}$.

Since $u_n$ is the oracle-use of $g(y;n)\de$ and $\tuple{n,s}\geq s$
we get $0^{u_n}\prec z$ and
\[
g(z; n)=g(0^{\omega}; n)=x(n)=0.
\]
This gives the contradiction $1=z(\tuple{n,s})=f(g(z),\tuple{n,s})=g(z;n)=0$.
\end{proof}

Our one-way function will be a permutation of selected bits of the input. We show that such maps meet
certain properties required of one-way functions.
\begin{lem}\label{VoqXmSWyYE}
If $p:\Nat\to\Nat$ is a computable injection then $f:\twome\to\twome$ with 
$f(x;n):=x(p(n))$ is a total computable random-preserving surjection.
\end{lem}\begin{proof}
Clearly $f$ is total computable. For each $y$ the real
\[
x(m):=\begin{cases}
y(n) & \textrm{if $m=p(n)$} \\
\ \ 0 & \rm{otherwise.}
\end{cases}
 \]
satisfies $f(x)=y$. So $f$ is surjective.
Let  $(V_i)$ be a universal \ml test with \pf and uniformly \ce members $V_i\subseteq\twomel$.
Then 
\[
f\inv(\dbra{V_i})=\bigcup_{\tau\in V_i} f\inv(\dbra{\tau}).
\]
Since $f$ is computable the sets $f\inv(\dbra{V_i})$ are uniformly \szn. Also
\[
f^{-1}(\dbra{\tau})=\sqbrad{x}{\forall i<|\tau|,\ x(p(i))=\tau(i)}
\]
and since $p$ is injective, $\mu(f\inv(\dbra{\tau}))=2^{-|\tau|}=\mu(\dbra{\tau})$. So
\[
\mu\parb{f\inv(\dbra{V_i})}\leq \sum_{\tau\in V_i}  \mu\parb{f\inv(\dbra{\tau})} =
\sum_{\tau\in V_i}  \mu\parb{\dbra{\tau}} =\mu(V_i).
\]
and $(f\inv(\dbra{V_i}))$ is a \ml test.
Since the reals $f$-mapping into $V_i$
are in $f\inv(\dbra{V_i})$ and $(V_i)$ is universal, every real $f$-mapping to a non-random real is non-random.
So $f$ is random-preserving.
\end{proof}

We extend the argument in  Proposition \ref{9w41OY8XJy} to obtain a one-way function.

\begin{thm}\label{649e64aa}
There is a total computable  random-preserving one-way surjection, such that every randomized inversion of it computes $\zj$. 
\end{thm}\begin{proof}
Let $f(x;\tuple{n,s}):=x(p(\tuple{n,s}))$ where
\[
p(\tuple{n,s}):=\begin{cases}
\ \ \ \ 2n &  \textrm{if $n\in \zj_s-\zj_{s-1}$} \\[0.13cm]
\ 2\tuple{n,s}+1 & \rm{otherwise}
\end{cases}
\]
so $f$ is total computable and 
\[
f(x;\tuple{n,s}):=\begin{cases}
\ \ \ \ x(2n) & \textrm{if $n\in \zj_s-\zj_{s-1}$} \\[0.13cm]
\ x(2\tuple{n,s}+1) & \rm{otherwise.}
\end{cases}
\]
Since $p$ is a computable injection, Lemma \ref{VoqXmSWyYE}
implies that $f$ is a total computable random-preserving surjection. 
Given $g:\subseteqm\twome\times\twome\to\twome$ and
 \[
 L:=\sqbrad{y\oplus r}{f(g(y,r))=y}
\hspace{0.3cm}\textrm{with \ $\mu(L)>0$}
 \]
it remains to  show that $g$ computes $\emptyset'$. Fix $\sigma$ such that 
\begin{equation}\label{da08d1d9}
\mu(L\cap\llbracket\sigma\rrbracket)>\frac{3}{4}\cdot \mu(\llbracket\sigma\rrbracket)
\end{equation} 
which exists by the Lebesgue density theorem \cite[Theorem 1.2.3]{rodenisbook}  
or by simply taking an open cover $\llbracket V\rrbracket$ of $L$ such that $\mu(\llbracket V\rrbracket-L)<\mu(L)/3$, 
where $V=\{\sigma_0,\sigma_1,\dots\}$ is prefix-free. Then some $\sigma_i$ in $V$ must satisfy \eqref{da08d1d9}. 

We now compute $\emptyset'$ using $g$ and the effective enumeration of a set $W$. 

To decide if $n\in\emptyset'$, we simultaneously compute $g(y,r;2n)$ for all $y,r$:
\[
\textrm{if $g(y,r;2n)\de$ with oracle-use $u\geq |\sigma|$ enumerate $(y\oplus r)\upharpoonright_u$ in $W$.}
\]
Then $L\subseteq\llbracket W\rrbracket$ and $W$ is a c.e.\ \pf  set with  $\forall \tau\in W,\ |\tau|\geq|\sigma|$.

Effectively in $n$ we  produce a computable enumeration $(W_s)$ of $W$ and 
\begin{itemize}
\item compute the least $t$ with $\mu(\llbracket W_t\rrbracket\cap\llbracket\sigma\rrbracket)> \mu(\llbracket\sigma\rrbracket)/2$
\item compute the length $k$ of the longest string in $W_t$
\end{itemize}
which exist by \eqref{da08d1d9} and  $L\subseteq \llbracket W\rrbracket=\bigcup_t \llbracket W_t\rrbracket$. 
It remains to show that 
\begin{equation}\label{TnoTtYPJlU}
n\in\emptyset'\iff n\in\emptyset'_k.
\end{equation}
For a contradiction assume that $n\in\emptyset'_s-\emptyset'_{s-1}$ for some $s>k$.

By the definition of $f$ we have $\forall x,\ f(x;\langle n,s\rangle)=x(2n)$  so
\[
y\oplus r\in L\implies
y(\langle n,s\rangle)=f(g(y\oplus r);\langle n,s\rangle)=g(y\oplus r;2n).
\] 
Fix $\tau\in W_t$. Since $\langle n,s\rangle\geq s>k\geq|\tau|$ we have 
\begin{itemize}
\item $g(y,r;2n)$ outputs the same result $i$ for $y\oplus r\in\llbracket\tau\rrbracket$
\item only half of the $y\oplus r\in\llbracket\tau\rrbracket$ satisfy $y(\langle n,s\rangle)=i$
\end{itemize}
so half of $y\oplus r\in\llbracket\tau\rrbracket$ must be outside $L$. Formally:
\[
\mu(\llbracket\tau\rrbracket-L)\geq \mu(\llbracket\tau\rrbracket)/2
\] 
for each $\tau\in W_t$. So $\mu(\llbracket\sigma\rrbracket-L)$ is at least
\[
\mu(\llbracket W_t\rrbracket\cap\llbracket\sigma\rrbracket-L)\geq 
\mu(\llbracket W_t\rrbracket\cap\llbracket\sigma\rrbracket)/2>\mu(\llbracket\sigma\rrbracket)/4
\] 
which contradicts \eqref{da08d1d9}, completing the
proof of \eqref{TnoTtYPJlU}.  So $g$ computes $\emptyset'$. 

Since no partial computable $g$ computes $\emptyset'$, $f$ is one-way.
\end{proof}

one-way functions are not probabilistically invertible on any sufficiently random $z$.
We quantify the level of randomness required for this fact.

\begin{prop}\label{8rIfjZG4Ef}
There is a random-preserving computable  $f:\twome\to\twome$ such that  
for each $z$, $g:\subseteqm\twome\times\twome\to\twome$  satisfying one of:
\begin{enumerate}[(i)]
\item  $z$ is  weakly 2-random  and $g$ is partial computable 
\item $z$ is  weakly 1-random  and $g$ is total computable  
\end{enumerate}
the probability that $g$ inverts $f$ on $z$ is 0.
\end{prop}\begin{proof}
Let $f,  L$ be as in Theorem \ref{649e64aa} and define
\[
B_q\ :=\ \sqbrad{y}{\mu(\sqbrad{r}{y\oplus r\in L})\geq q}.
\]
Then $\mu(L)=\mu(B_q)=0$ for each $q\geq 0$. Let $z$ be such that: 
\[
\mu(\sqbrad{r}{z\oplus r\in L})\geq q>0
\]
for a rational $q>0$ so $z\in B_q$.
If $g$ is partial computable then 
\[
y\oplus r\in L\iff\forall n\ \exists s\ f(g(y,r);n)[s]\downarrow=y(n)
\] 
so $L$ is $\Pi^0_2$. 
Let $(L_{n,s})$ be a computable family of clopen sets with 
\[
L=\bigcap_n\bigcup_s L_{n,s}
\hspace{0.3cm}\textrm{and}\hspace{0.3cm}
\bigcup_s L_{n+1,s}\subseteq\bigcup_s L_{n,s}
\]
and $L_{n,s}\subseteq L_{n,s+1}$. 
Then $B_q$ is a null $\Pi^0_2$ class as it is definable by
\[
y\in B_q\iff\forall n\ \exists s\ \mu(\{r:\{y\oplus r\in L_{n,s}\})\geq q.
\] 
By \eqref{VbNXoA2gzg} it follows that $z$ is not weakly 2-random.
If $g$ is total computable, then $L$ is a null $\Pi^0_1$ class.
By \eqref{VbNXoA2gzg} 
it follows that $z$ is not weakly 1-random.
\end{proof}

\subsection{Properties of one-way functions}\label{IDdXsdgbYm}
The proof that the total computable $f$ of \S\ref{W6lu6GL2ui} is one-way relied on  the fact that $f$ is not injective. 
This is not a coincidence: since  
\begin{itemize}
\item for each tree $T$ with  $[T]=\{x\}$ we have $x\leqT T$ uniformly in $T$
\item if $f$ is a total computable injection then $f^{-1}(y)$ consists of the unique path through a  tree
which is uniformly computable in $y$
\end{itemize}
total computable injections have  computable inverses. 
More generally:

\begin{thm}\label{BZwElndruD}
If $f:\twome\to\twome$ is  total computable, then there exists  
partial computable $g:\subseteqm\twome\to\twome$  which inverts $f$ 
 in $E:=\sqbrad{y}{\abs{f\inv(y)}=1}$.
\end{thm}
\begin{proof}
Let $\hat f$ be a computable representation of $f$. For  $y\in E$ the tree
\[
T_y=\{\sigma:\hat{f}(\sigma)\prec y\}
\] 
is uniformly computable in $y$: the map  $y\mapsto T_y$ is computable. Also
\begin{itemize}
\item $[T_y]=f^{-1}(y)$ because $f$ is total 
\item if $y\in E$ then  $f^{-1}(y)$ is a singleton so $T_y$ has a unique path.
\end{itemize}
The  path of  any  tree $T$ with $\abs{[T]}=1$  is computable uniformly in $T$.

So we can compute the unique $x$ in $f^{-1}(y)$ from $y\in E$ uniformly in $y$. 
\end{proof}
\begin{coro}\label{a50fb867a}
 Every total computable injection $f:\twome\to\twome$ has a total computable inversion 
 $g:\twome\to\twome$. 
\end{coro}

Partial computable injections need not have computable inverses.
\begin{thm}
There exists a partial computable injection $f:\subseteqm\twome\to\twome$ such that any inversion of $f$ computes $\zj$.
\end{thm}\begin{proof}
Let $f:\subseteqm\twome\to\twome$  be defined by $f(x):=p(x)\oplus q(x)$ where 
\begin{align*}
p(x; \tuple{n,s}):=&\begin{cases}
x(n) &\textrm{if $n\in\zj_s-\zj_{s-1}$}\\[0.1cm]
\ \ 0&\textrm{otherwise}
\end{cases}\\[0.1cm]
q(x; n):=&\begin{cases}
0 &\textrm{if $\forall i\leq n,\ \parb{x(i)=0\ \vee\  i\in\zj}$}\\[0.1cm]
\un &\textrm{otherwise.}
\end{cases}
\end{align*}
Then $f$ is partial computable and
\begin{itemize}
\item $f(x)\de$ iff $x$ (as a set of natural numbers)  is a subset of $\zj$
\item if $f(x)\de$, $f(z)\de$, $x\neq z$ then $x,z$ can only differ on positions in $\zj$.
\end{itemize}
In the latter case $p(x)\neq p(z)$ and $f(x)\neq f(z)$, so $f$ is injective.

As in the proof of Proposition \ref{9w41OY8XJy}, every inversion $g$ of $f$ computes $\zj$.
%
%
\end{proof}
 We now consider the extent to which total computable one-way functions fail to be injective.
Our results hold for a weaker type of one-way functions. 
\begin{defi}[Weakly one-way]
Given $f, h:\subseteqm\twome\to\twome$ let 
\[
L_{f}(h):=\sqbrad{y}{h(y)\de\wedga f(h(y))\de=y}.
\]
A partial computable $f:\subseteqm\twome\to\twome$ is {\em weakly one-way}
if $\mu(f(\twome))>0$ and $\mu(L_{f}(h))=0$ for each  partial computable $h:\subseteqm\twome\to\twome$.
\end{defi}
Clearly one-way functions are weakly one-way.

%
\begin{coro}\label{SXcsEXsjOj}
Every (weakly) one-way total computable $f:\twome\to\twome$  is almost nowhere injective.
\end{coro}\begin{proof}
Given $f$ as in the statement suppose that 
\[
\mu(E_f)>0
\hspace{0.3cm}\textrm{where}\hspace{0.3cm}
E_f:=\sqbrad{y}{\abs{f\inv(y)}=1}.
\]
By Theorem \ref{BZwElndruD} there is a partial computable function $h$ with domain $E_f$
and $\forall y\in E_f,\ f(h(y))=y$.  So $f$ is not weakly one-way.
\end{proof}
For the following, note that Theorem \ref{BZwElndruD} relativizes to any cylinder $\dbra{\sigma}$.

\begin{thm}\label{9j2ZZuHXST}
If $f:\twome\to\twome$ is a total computable (weakly) one-way function  then 
$f\inv(y)$ is uncountable for almost every $y\in f(\twome)$.
\end{thm}\begin{proof}
Suppose that  $f:\twome\to\twome$ is total computable.
Since  countable sets of reals are not perfect it suffices to show that if 
\[
\mu(D)>0
\hspace{0.3cm}\textrm{where}\hspace{0.2cm}
D:=\sqbrad{y}{\textrm{$f\inv(y)$ is not perfect}}
\]
then $f$ is not weakly one-way. 
For each $\sigma$ let 
\begin{itemize}
\item $f_{\sigma}$ be the restriction of $f$ to $\dbra{\sigma}$ 
and set $E_{\sigma}:=\sqbrad{y}{\abs{f_{\sigma}\inv(y)}=1}$
\item $g_{\sigma}$ be partial computable which inverts $f_{\sigma}$ on  $E_{\sigma}$
\end{itemize}
where the existence of the $g_{\sigma}$ follows from Theorem \ref{BZwElndruD}.

If $y\in D$ the closed set $f\inv(y)$ has an isolated path, so
\[
\exists \rho,\ y\in E_{\rho}
\hspace{0.3cm}\textrm{and}\hspace{0.3cm}
D\subseteq\bigcup_{\sigma} E_{\sigma}.
\]
Since  $\mu(D)>0$ there is $\sigma$ with $\mu(E_{\sigma})>0$ and
by the choice of $g_{\sigma}$: 
\[
\mu(\sqbrad{y}{g_{\sigma}(y)\de\wedga f(g_{\sigma}(y))=y})>0.
\]
This shows that $f$ is not weakly one-way.
\end{proof}

\section{Inversions of nearly injective functions}\label{3p5CxCWq7F}
Since one-way injections are well-studied  in computational complexity 
\cite{onewayPermImp, onewaypermRothe} it is interesting to ask if there are one-way injections $f$. 
Their existence remains unknown but as discussed in \S\ref{IDdXsdgbYm}, they cannot be total.

With this motivation, we examine the extent to which non-injectivity is essential in the arguments  of
\S\ref{4GY8Jx1fU}. 
We exhibit total computable random-preserving surjections
that are hard to invert and are {\em nearly injective}. 
\begin{defi}
We say that $f:\twome\to\twome$ is {\em two-to-one} if $\forall y,\ \abs{f\inv(y)}\leq 2$.
\end{defi}
By extending the method of  \S\ref{4GY8Jx1fU}, 
in \S\ref{2cPoDTVyqX} we 
exhibit a  total computable two-to-one  random-preserving surjection
whose inversions compute $\zj$ but is nevertheless almost everywhere effectively invertible. 
In \S\ref{9TULCemTNb} we  exhibit 
a two-to-one total computable random-preserving  surjection
with no effective probabilistic map that inverts it  almost everywhere.

\subsection{Blueprint for two-to-one functions}\label{2cPoDTVyqX}
Our maps will be of the form $f(x\oplus z):=h^z(x)\oplus z=y\oplus z$ where
\begin{itemize}
\item $h^z$ selects positions of $x$-bits {\em used} in (\ie copied into) $h^z(x)$ 
\item all but at most one position are {\em used} in $h^z(x)$.
\end{itemize}

\begin{figure}
\mybox{\parbox{11cm}{{\em Input:} $x\oplus z$; \hspace{0.1cm} {\em Output:} $y$ with $f(x\oplus z)=y\oplus z$.} 
\begin{algorithmic}[1]
\STATE \vspace{-0.3cm} Initialization: $k:=0$, $s:=0$
\WHILE{true}
\IF{$k\in\zj_s$ or $E_s^z(k)$}
\STATE $y(s):=x(k)$
\STATE $k:=s+1$
\ELSE
\STATE $y(s):=x(s+1)$
\ENDIF
\STATE $s:=s+1$
\ENDWHILE
\end{algorithmic}}
\caption{Definition of $f$ given a $z$-computable predicate $E^z_{s}(i)$.}\label{KGNpueO7aa}\vspace{-0.3cm}
\end{figure}
The selection of $x$-bits is facilitated by a movable marker $k$ and depends on a computable predicate $E^z_{s}(i)$ as shown in
Figure \ref{KGNpueO7aa}. Let $k^z_s$ be the candidate for the unique {\em unused} position at $s$,  which corresponds to $k$ in Figure \ref{KGNpueO7aa}. 

The update of $k^z_s$   occurs if
\begin{itemize}
\item either $k^z_s\in \zj_s$ which we call a {\em $\zj$-permission}
\item or $E^z_s(k^z_s)$ which we call a {\em $z$-permission}
\end{itemize}
in which case candidate $k^z_s$ is eliminated. Since $k^z_s$ is non-decreasing:
\begin{enumerate}[(i)]
\item  if $\lim_s k^z_s=\infty$   all $x$-positions are {\em used} in $h^z(x)$ 
\item if $\lim_s k^z_s=k_0$  all $x$-positions except $k_0$   are {\em used} in $h^z(x)$.
\end{enumerate}
So $\abs{f\inv(h^z(x)\oplus z)}=1$ if  (i) holds and $\abs{f\inv(h^z(x)\oplus z)}=2$ if (ii) holds.

We now need a relativization of Lemma \ref{VoqXmSWyYE} that applies to this
extended form of selective permutation. To this end, we use a fact from \cite{vLamb90}:
\begin{equation}\label{vyx2njiUf5}
\textrm{$x\oplus y$ is random iff $x$ is random and $y$ is $x$-random}
\end{equation}
also known as {\em van Lambalgen's theorem}.
\begin{lem}\label{VoqXmSWyYErel}
Let $p$ be a total Turing functional such that
$n\mapsto p^z(n)$ is injective for each oracle  $z$ and define
$f, h^z:\twome\to\twome$ by
\[
h^z(x;n):= x(p^z(n))
\hspace{0.3cm}\textrm{and}\hspace{0.3cm}
f(x\oplus z):=h^z(x)\oplus z.
\]
Then  $f$ is a total computable  random-preserving surjection.
\end{lem}\begin{proof}
Since $p$ is a total Turing functional, $(z,x)\mapsto h^z(x)$ and $f$ are computable.
By the  relativization of  Lemma \ref{VoqXmSWyYE} to arbitrary oracle $z$
it follows that $g^z$ is a $z$-random preserving surjection.

So $f$ is a surjection. If $x\oplus z$ is random, by \eqref{vyx2njiUf5} we get that
 \[
 \textrm{ $x$ is $z$-random}\impl
 \textrm{$h^z(x)$ is $z$-random} 
 \]
 so $h^z(x)\oplus z$ is random. Hence $f$ is random-preserving.
\end{proof}

We now apply the above framework to prove:

\begin{thm}\label{2bDADUaw4x}
There is a total computable  $f:\twome\to\twome$ such that:
\begin{itemize}
\item $f$ is a two-to-one random-preserving surjection
\item $f$ is almost everywhere effectively invertible
\item every $g:\subseteqm\twome\to\twome$ that inverts $f$  computes $\zj$
\end{itemize}
and the latter  holds for the restriction of $f$ in any cylinder $\dbra{\sigma}$.
\end{thm}\begin{proof}
We define  $f$ as above with predicate $z(\tuple{i,s})=1$  in place of $E^z_s(i)$. 

For each $z$ let $k^z_0=0$ and 
\begin{equation}\label{xzycPQkSMH}
k^z_{s+1}:=\begin{cases}
s+1 &  \textrm{if $k^z_{s}\in \zj_{s}$ or $z(\tuple{k^z_{s},s})=1$} \\[0.13cm]
\ \ k^z_{s} &  \textrm{otherwise.}
\end{cases}
\end{equation}
To select the next $x$-bit used in $f(x\oplus z)$ define:
\[
p^z_s:=\begin{cases}
s+1 &  \textrm{if $k^z_{s+1}=k^z_{s}$} \\[0.13cm]
\ \ k^z_{s} &  \textrm{otherwise.}
\end{cases}
\]
Then $s\mapsto p^z_s$ is injective for each $z$. By  Lemma \ref{VoqXmSWyYErel} the $f$  given by
\[
f(x\oplus z):=h^z(x)\oplus z
\hspace{0.3cm}\textrm{where}\hspace{0.3cm}
h^z(x; s):=x(p^z_s)
\]
is a total computable  random-preserving surjection.
Also 
\[
\forall k,\ \abs{\sqbrad{i}{z(\tuple{k, i})=1}}=\infty\impl \lim_s k^z_s=\infty\impl \abs{f\inv(h^z(x)\oplus z)}=1
\]    
as discussed above, and $f$ is two-to-one.
This condition is met for all random $z$, so $f$ is almost everywhere injective. By Theorem \ref{BZwElndruD} there is a 
partial computable almost everywhere inversion of $f$.

Assuming that $g:\subseteqm\twome\to\twome$ inverts $f$ we show that $g\geqT \zj$. 

To decide if $n\in\zj$, we define $z$ so that  $k^z_t$  gets  stuck on $n$ unless $n\in \zj$:
\begin{equation}\label{dpDCqn3bGS}
z(\langle i,s\rangle):=\begin{cases}
0 & \text{if } i=n \\
1 & \text{if }  i\neq n.
 \end{cases}
\end{equation}
Then $k_t^{z}$ does not get  stuck on any number $\neq n$ in the sense that
\begin{enumerate}[\hthree(a)]
\item  $k_t^z=n$ at $t=n$
\item  $k_t^z$ is updated at $t>n$ iff $n\in\zj_t$.
\end{enumerate}
We show  $n\in\zj$ iff $n\in\zj_{\max\{u,s\}}$,  where
$u$ is the oracle-use of $g(\zerome\oplus z; 2n)$.

For a contradiction suppose  $n\in\zj_t-\zj_{t-1}$ for  $t>\max\{u,s\}$ so
\begin{equation}\label{vu4ROYq1tU}
\forall x,y\ \ \parb{f(x\oplus z)=y\oplus z\impl y(t)=x(n)}
\end{equation}
due to (b) and the definition of $f$. Let $y_0:=\zerome $ and $y_1:=0^{t}10^\omega$.

Then for $i=0,1$ we have $f(g(y_i \oplus z))=y_i \oplus z$ and by \eqref{vu4ROYq1tU}: 
\begin{equation}\label{h5gCeDoGIw}
g(y_0\oplus z; 2n)=y_0(t)=0
\hspace{0.3cm}\textrm{and}\hspace{0.3cm}
g(y_1\oplus z; 2n)=y_1(t)=1.
\end{equation}
Since $u$ is the oracle-use of $g(y_0 \oplus z_n; 2n)$ and $t>u$ we have
\[
(y_0\oplus z)\restr_{u} \prec y_1\oplus z
\hspace{0.3cm}\textrm{so}\hspace{0.3cm}
g(y_0\oplus z; 2n)=g(y_1\oplus z; 2n)
\]
which contradicts \eqref{h5gCeDoGIw}.
It follows that $n\in\zj\iff n\in\zj_{\max\{u,s\}}$.

Finally we modify the above argument so that  $g\geqT \zj$ is obtained from  the weaker assumption that 
$g$ inverts $f$ inside a cylinder  $\dbra{\upsilon \oplus \zeta}$. Let
\begin{itemize}
\item  $z$ be the extension of $\zeta$ given by  \eqref{dpDCqn3bGS} for $\tuple{i,s}\geq  |\zeta|$
\item $u$ be the oracle-use of $g(\upsilon \zerome\oplus z; 2n)\de$
\item $y_0:=\upsilon\zerome$ and $y_1:=\upsilon 0^{t-|\upsilon|}10^\omega$
\end{itemize}
for $n> |\zeta|$ and $t>u$ which are used for deciding if $n\in\zj$ as before.

Assuming $n> |\zeta|$  the modified $z$ satisfies (a), (b). So the above argument
applies to the modified $y_0, y_1, u$ and proves $g\geqT\zj$ as required.
\end{proof}

\subsection{Almost everywhere probabilistic inversions}\label{9TULCemTNb}
By Theorem \ref{9j2ZZuHXST} every computable two-to-one  random-preserving
surjection can be effectively inverted with positive probability.
This leaves the possibility that a total computable  $f$ exists such that:
\begin{enumerate}[(i)]
\item $f:\twome\to\twome$ is two-to-one, random-preserving and surjective
\item no partial computable $g$ inverts $f$ with probability 1.
\end{enumerate}
We construct $f$ with the above properties within the framework of \S\ref{2cPoDTVyqX}, starting
with the modifications and additional ideas needed to achieve this.

In \S\ref{2cPoDTVyqX} we relied on the fact that the given candidate $g$ for inverting $f$
was defined on certain specially constructed computable reals $y_i\oplus z_n$.
This may no longer be the case since we can only assume that $g$ is defined on a set of measure 1.
We  restrict our considerations to sufficiently random reals.

The domain of a partial computable $g$ which is defined almost everywhere is a $\Pi^0_2$ class of measure 1 and
includes all weakly randoms. In general we only have $g\not\geqT\zj$ so for some $r\not\geqT \zj$ we use {\em weakly $r$-randoms}:
reals that are members of every $\Sigma^0_1(r)$ class of measure 1.

\begin{lem}\label{EVZuKrGCnN}
Suppose that $f:\twome\to\twome$ is a computable surjection and 
$g:\subseteqm \twome\to\twome$ is an almost everywhere inversion of $f$. 
If $g\leqT r$ then $g$ inverts $f$ on each weakly $r$-random real.
\end{lem}\begin{proof}
The set of reals where $g$ inverts $f$ is the $\Pi^0_2(r)$ class
\[
L_f(g):=
\sqbrad{y}{g(y)\de\wedga f(g(y))=y}
\]
which has measure 1 according to the hypothesis.
Since weakly $r$-random reals belong to every $\Sigma^0_1(r)$ class of measure 1,
they also belong to every $\Pi^0_2(r)$ class of measure 1. So $L_f(g)$ contains every
weakly $r$-random real.
\end{proof}
At this point we need  some facts about {\em generic reals}. 

Genericity is a topological form of {\em typicality} which
has been extensively studied along with randomness in 
computability \cite{Jockusch80, typical} and computational 
complexity \cite{WANG200033, S009753979630235X, 829497829767, LORENTZ1998245}.
Generic reals avoid  every  {\em definable} meager (as in Baire category) set, as algorithmically random reals avoid  
definable null set. The level of {\em definability} determines the strength of genericity or randomness.  

\begin{defi}[\citet{Jockusch80}]
Given $r, w\in\twome$, if  
\[
\exists \sigma\prec w: \ \parb{\dbra{\sigma}\subseteq G\ \vee\ \dbra{\sigma}\cap G=\emptyset}.
\]
 for every $\Sigma^0_1(r)$ class $G\subseteq\twome$ we say that  $w$ is {\em $r$-generic}.
\end{defi}
By \cite{Jockusch80} and \cite{Kurtz81} (see \cite[Theorem 8.11.7]{rodenisbook}) respectively, for each $r$:
\begin{enumerate}[\hspace{0.3cm}(a)]
\item  if  $r\not\geqT\zj$ and $z$ is $r$-generic then $z\oplus r\not\geqT \zj$
\item every $r$-generic is weakly $r$-random and not  random.
\end{enumerate}
\begin{defi}
The {\em $n$th column of  $w$} is the real  $w^n$  with $w^n(i):=w(\tuple{n,i})$.
\end{defi}
By the analogue of  \eqref{vyx2njiUf5}  for genericity \cite[Proposition 2.2]{Yuliang05} we have
\begin{equation}\label{fHW6C3KBo3}
\textrm{if $w$ is $r$-generic then $w^n$ is $r$-generic for each $n$.}
\end{equation}
By \cite[Corollary 2.1]{BDNGP} one implication of \eqref{vyx2njiUf5} holds for weak randomness:
\begin{equation}\label{fHW6C3KBo3as}
\parbox{9cm}{if $w$ is weakly $r$-random and $y$ is weakly $w\oplus r$-random  then $y\oplus w$ is weakly $r$-random.}
\end{equation}
Recall that  incomputable sets cannot be computed probabilistically with non-zero probability \cite{deleeuw1955}.
A  relativization of this fact is  
\begin{equation}\label{le9rALzZin}
y\not\geqT r\implies \mu\parb{\sqbrad{x}{x\oplus y\geqT r}}=0
\end{equation}
which can be found in \cite[Corollary 8.12.2]{rodenisbook}.

Given $r\geqT g$, we  adapt the argument of \S\ref{2cPoDTVyqX} 
by choosing $y_i, z$ so that the $y_i\oplus z$ are weakly $r$-random.
We construct them from the $y,w$ given by:
\begin{lem}\label{YCjtrEhy3n}
For each $r\not\geqT\zj$ there exist $y, w$ such that 
\begin{enumerate}[(i)]
\item $y\oplus  w$ is weakly $r$-random and $y$ is  random 
\item  no column $w^n$ of $w$ is  random
\item  $r\oplus y\oplus  w\not\geqT\zj$.
\end{enumerate}
\end{lem}\begin{proof}
Let $w$ be $r$-generic and $y$ be $(w\oplus r)'$-random so 
\begin{equation}\label{jMzAEBN4oY}
\textrm{$y$ is weakly 2-random relative to $w\oplus r$}
\end{equation}
and by (b), $w$ is weakly $r$-random. 
This, combined with \eqref{fHW6C3KBo3as}, gives (i).

By \eqref{fHW6C3KBo3} each $w^n$ is $r$-generic and by (b) not random, hence (ii).

By (a) we have $w\oplus r\not\geqT \zj$ so by \eqref{le9rALzZin} the $\Sigma^0_3(w\oplus r)$ class
\[
G:=\sqbrad{x}{x\oplus w\oplus r\geqT\zj}
\]
is null. By \eqref{jMzAEBN4oY} and \eqref{VbNXoA2gzg} we get
$y\not\in G$ and (iii) holds.
\end{proof}

Fix $w,y$ as in  Lemma \ref{YCjtrEhy3n} and let 
$(U_s)$ be an effective enumeration of a member $U$ 
of a universal \ml test with $y\not\in U$ and $\forall n,\ w^n\in U$. 

For each $n$ we $(y \oplus w)$-effectively define a real $z$ consisting of the columns of $w$ except for the $n$th column which is $y$. 
Then $y\oplus w$ is weakly $r$-random by 
\eqref{fHW6C3KBo3as}. Since $z$ copies $y\oplus w$ from a computable array of indices,  $z$ is also weakly $r$-random. 
This suggests defining $E^z_{s}(i)$ in the template of \S\ref{2cPoDTVyqX}  in terms of the memberships $z^i\in U$ of the $i$th column $z^i$ of $z$.

A last hurdle in the adaptation of  \S\ref{2cPoDTVyqX} is the requirement  
\begin{equation}\label{XHqigBAD9N}
\textrm{from $w\oplus y, n$ compute $s_n$ such that $k^{z}_{s_n}$ is {\em used} iff $n\not\in\zj$}
\end{equation}
so  $k^{z}_{s_n}$ does not get $z$-permission.
In \S\ref{2cPoDTVyqX} permissions of $k_s^z$ at $s$ depended entirely on the value of $k_s^z$ so  
\eqref{XHqigBAD9N} was achieved by defining $s_n, z$ with $k_{s_n}^{z}=n$. 
This is no longer possible as  we do not have control over the stages $s$ where the $w^i$ appear in $U$.
The solution is to define permissions in terms of
\[
d^z_s:=\abs{\sqbrad{k^z_t}{t\leq s}}
\]
instead of $k_s^z$. This will allows to define the required $s_n, z$ for \eqref{XHqigBAD9N}.
 
{\bf Parameters.}
Given $z$ let  $k^z_s$ be the non-decreasing {\em counter} with  $k^z_0=0$ and 
\begin{equation*}
k^z_{s+1}:=\begin{cases}
s+1 &  \textrm{if $d^z_{s}\in \zj_{s}$ or $z^{d^z_{s}}\in U_s$} \\[0.13cm]
\ \ k^z_{s} &  \textrm{otherwise}
\end{cases}
\end{equation*}
where $d^z_{s}:=\abs{\sqbrad{t<s}{k^z_{t+1}\neq k^z_{t}}}$ counts the updates of $k^z$ and
\[
p^z_s:=\begin{cases}
s+1 &  \textrm{if $k^z_{s+1}=k^z_{s}$} \\[0.13cm]
\ \ k^z_{s} &  \textrm{otherwise}
\end{cases}
\]
enumerates $\Nat$, omitting $k^z_{s}$ as long as it is not updated and including it otherwise. 
Updates of  $k^z$ coincide with those of $d^z$ and are due to one of:
\begin{itemize}
\item $d^z_{s}\in \zj_{s}$ which we call  {\em $\zj$-permission} of $k^z_{s}$ at $s+1$
\item $z^{d^z_{s}}\in U_s$ which we call  {\em $z$-permission} of $k^z_{s}$ at $s+1$.
\end{itemize}
We say that $k^z_{s}$ {\em receives permission at $s+1$} if one of the above clauses hold.

Both $d^z, k^z$ are non-decreasing and $d^z$ increases by at most 1.
The reason that  $k^z$ is allowed to skip numbers is so that the range of $p^z$ misses at most one number $m$,
which happens exactly when $\lim_s k^z_s=m$.

So  $k^z_s$ becomes some $p^z_t$ iff  $d^z_s$ receives permission at a stage $t>s$.
\begin{lem}\label{CNpHUM8nIK}
Given $r\not\geqT \zj$ let  $y,  w$ be as in Lemma \ref{YCjtrEhy3n}.
Effectively in $y\oplus w$ and  $n$ we can  define $z$ and  $s$ such that
\begin{enumerate}[(i)]
\item $d^z_{s}=n$ and  $\parb{\lim_t k^z_t<\infty \iff \lim_t k^z_t=k^z_s \iff n\not\in\zj}$
\item $y\oplus z$ is weakly $r$-random.
\end{enumerate}
\end{lem}\begin{proof}
Let $z$ be the real obtained from $w$ by replacing its $n$th column by $y$.
Since $\forall d,\ w^d\in U$ each $d\neq n$ will receive $z$-permission at some stage.
If $s$ is the stage where each $d<n$ have received permission, 
$d^z_{s}=n$. Since $z^n=y\not\in U$ it follows that  $n$ will never receive $z$-permission.
So $d^z$ gets stuck at $n$ if and only if $n$ does not receive $\zj$-permission. 
The required equivalence then follows, given that each $d>n$ receives $z$-permission.
\end{proof}
We are now ready to prove:

\begin{thm}\label{K1nykJEktW}
There is a total computable  $f:\twome\to\twome$ such that:
\begin{itemize}
\item $f$ is a two-to-one random-preserving surjection 
\item for each partial $g\not\geqT \zj$, with positive probability, $g$ fails to invert $f$
\end{itemize}
and the latter  holds for the restriction of $f$ in any cylinder $\dbra{\sigma}$.
\end{thm}\begin{proof}
Define the computable $f:\twome\to\twome$ by
\[
f(x\oplus z):=h^z(x)\oplus z
\hspace{0.3cm}\textrm{where}\hspace{0.3cm}
h^z(x; s):=x(p^z_s)
\]
so $h^z(x)$ outputs the bits of $x$ in some order determined by $(\zj_s)$, $z$ 
with the exception of $\lim_s k^z_s$ when this limit is finite. Then (i), (ii) of \S\ref{2cPoDTVyqX}  hold
and $f$ is two-to-one.  Since $s\mapsto p^z_s$ is injective,  Lemma \ref{VoqXmSWyYErel} shows that 
$f$  is a computable random-preserving surjection. 

Assuming $g:\subseteqm\twome\to\twome$ inverts $f$ with probability 1, we show that $g \geqT \zj$. 

For a contradiction assume $g \not\geqT \zj$, fix
$r\not\geqT\zj$ with $g\leqT r$  and
\[
\textrm{fix $y_0, v, w$ as in Lemma \ref{YCjtrEhy3n} for $y=y_0$, so $y_0\oplus w \oplus r\not\geqT \zj$.}
\]
Fix $n$. To decide if $n\in\zj$, from $y_0\oplus w \oplus r$ we 
\begin{itemize}
\item effectively compute $z,s$ as in Lemma \ref{CNpHUM8nIK}
\item let $k:=k^z_s$ be the potential finite limit of $k^{z}$
\end{itemize}
so  $y_0 \oplus z$ is weakly $r$-random, $d^z_s=n$ and 
$\forall t>s\ \parb{n\in \zj_t\impl p^{z}_{t}=k}$.

By the definition of $f$, for each $x,y$: 
\begin{equation}\label{e13u5DhHKT}
\parb{f(x\oplus z)=y\oplus z\wedga t>s\wedga n\in \zj_t}\implies y(t)=x(k).
\end{equation}
Since $y_0 \oplus  z$ is weakly $r$-random, by Lemma \ref{EVZuKrGCnN}  
we have $g(y_0\oplus z)\de$.

Letting $u$ be the oracle-use of $g(y_0\oplus z; 2k)\de$ we claim that
\begin{equation}\label{RjR2pYxv2i}
n\in\zj\iff n\in\zj_{\max\{u,s\}}.
\end{equation}
Otherwise there is $t> \max\{u, s\}$ with $n\in\zj_t-\zj_{t-1}$. Let $y_1$ be the real with 
\[
y_1(i)=y_0(i)\iff i\neq t.
\]
Then $y_1 \oplus z$ is weakly $r$-random, so by Lemma \ref{EVZuKrGCnN}, for $i=0,1$
\begin{equation}\label{bvoiOcCWY3a}
g(y_i\oplus z)\de
\hspace{0.3cm}\textrm{and}\hspace{0.4cm}
 f(g(y_i\oplus z))=y_i\oplus z.
\end{equation}
Since $t>u$ we have $(y_0\oplus z)\restr_u \prec y_1\oplus z$ so by \eqref{RjR2pYxv2i} we get
\begin{equation}\label{h5gCeDoGIwa}
g(y_0\oplus z; 2k)=g(y_1\oplus z; 2k).
\end{equation}
By \eqref{bvoiOcCWY3a}, \eqref{e13u5DhHKT} and $t>s$ we get 
\[
g(y_0\oplus z; 2k)=y_0(t)
\hspace{0.3cm}\textrm{and}\hspace{0.3cm}
g(y_1\oplus z; 2k)=y_1(t)
\]
which contradict \eqref{h5gCeDoGIwa} since $y_0(t)\neq y_1(t)$. 
This concludes the proof of \eqref{RjR2pYxv2i}.

By \eqref{RjR2pYxv2i} we get $\zj\leqT y_0\oplus w\oplus r$ which contradicts
the hypothesis that $y_0\oplus w \oplus r\not\geqT \zj$.
We conclude that $g \geqT \zj$ so every almost everywhere inversion of $f$ computes $\zj$.
The same argument applies in the case that $g$ inverts $f$ almost everywhere in a  
cylinder $\dbra{\sigma}$.
\end{proof}

\section{Conclusion}\label{UWucNoldLo}
We constructed a one-way  function on the reals, which is the analogue of the one-way functions in computational complexity 
formulated by Levin in \cite{levin24ZF}. 
We argued that Levin's definition is the correct analogue of the one-way functions from computational complexity,
over weaker alternatives. An analysis of  alternatives was conducted in \cite{owrand24}, along with a study 
of the oracles needed to probabilistically invert one-way real functions.

Our result was based on a general 
framework for constructing permutations of selected bits of the input, 
which was adapted toward understanding the extent to which
one-way functions can be injective. 
Applications often require additional ideas, but can yield analogues of significant properties  
in computational complexity, such as collision-resistance \cite{siblibarmpZhang}. 

Despite the non-trivial adaptations that are often required, 
all currently known one-way real functions (as defined in \cite{levin24ZF}) 
are permutations of selected bits of the input and can thus be viewed as applications of our framework.
The question therefore arises with respect to the generality of this framework. This is  relevant to  the 
existence of a partial computable one-way injection, which is currently unknown.


\begin{thebibliography}{30}
\providecommand{\natexlab}[1]{#1}
\providecommand{\url}[1]{\texttt{#1}}
\expandafter\ifx\csname urlstyle\endcsname\relax
  \providecommand{\doi}[1]{doi: #1}\else
  \providecommand{\doi}{doi: \begingroup \urlstyle{rm}\Url}\fi

\bibitem[Ambos-Spies(1995)]{829497829767}
K.~Ambos-Spies.
\newblock Resource-bounded genericity.
\newblock In \emph{Proc. 10th Annual Structure in Complexity Theory Conference
  (SCT'95)}, SCT'95, USA, 1995. IEEE Computer Society.

\bibitem[Barmpalias and Zhang(2024)]{siblibarmpZhang}
G.~Barmpalias and X.~Zhang.
\newblock Collision-resistant hash-shuffles on the reals.
\newblock Arxiv 2501.02604, 2024.

\bibitem[Barmpalias et~al.(2011)Barmpalias, Downey, and Ng]{BDNGP}
G.~Barmpalias, R.~Downey, and K.~M. Ng.
\newblock Jump inversions inside effectively closed sets and applications to
  randomness.
\newblock \emph{J.\ Symb.\ Log.}, 76\penalty0 (2):\penalty0 491--518, 2011.

\bibitem[Barmpalias et~al.(2014)Barmpalias, Day, and Lewis-Pye]{typical}
G.~Barmpalias, A.~R. Day, and A.~E.~M. Lewis-Pye.
\newblock The typical {T}uring degree.
\newblock \emph{Proc. London Math. Soc.}, 109\penalty0 (1):\penalty0 1--39,
  2014.

\bibitem[Barmpalias et~al.(2024)Barmpalias, Wang, and Zhang]{owrand24}
G.~Barmpalias, M.~Wang, and X.~Zhang.
\newblock Complexity of inversion of functions on the reals.
\newblock Arxiv 2412.07592, 2024.

\bibitem[Barmpalias et~al.(December 2023)Barmpalias, G\'{a}cs, Levin,
  Lewis-Pye, and Shen]{LevinEmailDec23}
G.~Barmpalias, P.~G\'{a}cs, L.~Levin, A.~Lewis-Pye, and A.~Shen.
\newblock Email correspondence, December 2023.

\bibitem[de~Leeuw et~al.(1955)de~Leeuw, Moore, Shannon, and
  Shapiro]{deleeuw1955}
K.~de~Leeuw, E.~F. Moore, C.~E. Shannon, and N.~Shapiro.
\newblock {Computability by probabilistic machines}.
\newblock In C.~E. Shannon and J.~McCarthy, editors, \emph{Automata Studies},
  pages 183--212. Princeton University Press, Princeton, NJ, 1955.

\bibitem[Downey and Hirschfeldt(2010)]{rodenisbook}
R.~G. Downey and D.~Hirschfeldt.
\newblock \emph{Algorithmic Randomness and Complexity}.
\newblock Springer, 2010.

\bibitem[G{\'a}cs(May 8, 2024)]{Gacs24oneway}
P.~G{\'a}cs.
\newblock A (partially) computable map over infinite sequences can be
  `one-way'.
\newblock Privately circulated draft, May 8, 2024.

\bibitem[Gherardi(2011)]{Gherardi11T}
G.~Gherardi.
\newblock {Alan Turing} and the foundations of computable analysis.
\newblock \emph{Bull. Symbolic Logic}, 17\penalty0 (3):\penalty0 394--430,
  2011.

\bibitem[H{\AA}stad et~al.(1999)H{\AA}stad, Impagliazzo, Levin, and
  Luby]{Hstad1999}
J.~H{\AA}stad, R.~Impagliazzo, L.~A. Levin, and M.~Luby.
\newblock A pseudorandom generator from any one-way function.
\newblock \emph{SIAM J. Comput.}, 28\penalty0 (4):\penalty0 1364–1396, 1999.

\bibitem[Hemaspaandra and Rothe(2000)]{onewaypermRothe}
L.~A. Hemaspaandra and J.~Rothe.
\newblock Characterizing the existence of one-way permutations.
\newblock \emph{Theor. Comput. Sci.}, 244\penalty0 (1):\penalty0 257--261,
  2000.

\bibitem[Hirahara et~al.(2023)Hirahara, Ilango, Lu, Nanashima, and
  Oliveira]{Hirahara2023}
S.~Hirahara, R.~Ilango, Z.~Lu, M.~Nanashima, and I.~C. Oliveira.
\newblock A duality between one-way functions and average-case symmetry of
  information.
\newblock In \emph{Proceedings of the 55th Annual ACM Symposium on Theory of
  Computing}, STOC’23. ACM, 2023.

\bibitem[Impagliazzo and Levin(1990)]{Impagliazzo1990}
R.~Impagliazzo and L.~Levin.
\newblock No better ways to generate hard {NP} instances than picking uniformly
  at random.
\newblock In \emph{Proceedings [1990] 31st Annual Symposium on Foundations of
  Computer Science}. IEEE, 1990.

\bibitem[Impagliazzo and Luby(1989)]{Impagliazzo1989}
R.~Impagliazzo and M.~Luby.
\newblock One-way functions are essential for complexity based cryptography.
\newblock In \emph{30th Annual Symposium on Foundations of Computer Science}.
  IEEE, 1989.

\bibitem[Impagliazzo and Rudich(1989)]{onewayPermImp}
R.~Impagliazzo and S.~Rudich.
\newblock Limits on the provable consequences of one-way permutations.
\newblock In \emph{Proc. 21st Annu. ACM Symp. Theory Comput.}, STOC'89, New
  York, NY, USA, 1989. Assoc. Comput. Mach.

\bibitem[Jockusch(1980)]{Jockusch80}
C.~Jockusch, Jr.
\newblock Degrees of generic sets.
\newblock In F.~R. Drake and S.~S. Wainer, editors, \emph{Recursion Theory: Its
  Generalizations and Applications, Proceedings of Logic Colloquium '79, Leeds,
  August 1979}, pages 110--139, Cambridge, U. K., 1980. Cambridge University
  Press.

\bibitem[Kurtz(1981)]{Kurtz81}
S.~Kurtz.
\newblock \emph{Randomness and genericity in the degrees of unsolvability}.
\newblock Ph.{D.} {D}issertation, University of Illinois, Urbana, 1981.

\bibitem[Levin(2003)]{Levin2003}
L.~A. Levin.
\newblock The tale of one-way functions.
\newblock \emph{Probl. Inf. Transm.}, 39\penalty0 (1):\penalty0 92–103, 2003.

\bibitem[Levin(2022-2024)]{levin24ZF}
L.~A. Levin.
\newblock {Zermelo-Fraenkel Axioms, Internal Classes, External Sets}.
\newblock ArXiv 2209.07497, versions 1-15, 2022-2024.

\bibitem[Liu and Pass(2020)]{Liu2020}
Y.~Liu and R.~Pass.
\newblock On one-way functions and {K}olmogorov complexity.
\newblock In \emph{{2020 IEEE 61st Annual Symposium on Foundations of Computer
  Science (FOCS)}}. IEEE, 2020.

\bibitem[Liu and Pass(2023)]{Liu2023}
Y.~Liu and R.~Pass.
\newblock \emph{{One-Way Functions and the Hardness of (Probabilistic)
  Time-Bounded {K}olmogorov Complexity w.r.t. Samplable Distributions}}, page
  645–673.
\newblock Springer Nature Switzerland, 2023.

\bibitem[Lorentz and Lutz(1998)]{LORENTZ1998245}
A.~K. Lorentz and J.~H. Lutz.
\newblock Genericity and randomness over feasible probability measures.
\newblock \emph{Theor. Comput. Sci.}, 207\penalty0 (1):\penalty0 245--259,
  1998.

\bibitem[Miller(2004)]{Miller2004}
J.~S. Miller.
\newblock Degrees of unsolvability of continuous functions.
\newblock \emph{J.\ Symb.\ Log.}, 69\penalty0 (2):\penalty0 555–584, 2004.

\bibitem[Pour-El and Richards(1989)]{PR89}
M.~B. Pour-El and J.~I. Richards.
\newblock \emph{Computability in Analysis and Physics}.
\newblock Perspectives in Mathematical Logic. Springer, Berlin, 1989.

\bibitem[Segev(2023)]{Segev2023}
G.~Segev.
\newblock Finding connections between one-way functions and {K}olmogorov
  complexity.
\newblock \emph{{Commun. {ACM}}}, 66\penalty0 (5):\penalty0 90–90, 2023.

\bibitem[van Lambalgen(1990)]{vLamb90}
M.~van Lambalgen.
\newblock The axiomatization of randomness.
\newblock \emph{J. Symbolic Logic}, 55\penalty0 (3):\penalty0 1143--1167, 1990.

\bibitem[Wang(1998)]{S009753979630235X}
Y.~Wang.
\newblock Genericity, randomness, and polynomial-time approximations.
\newblock \emph{SIAM J. Comput.}, 28\penalty0 (2):\penalty0 394--408, 1998.

\bibitem[Wang(2000)]{WANG200033}
Y.~Wang.
\newblock Resource bounded randomness and computational complexity.
\newblock \emph{Theor. Comput. Sci.}, 237\penalty0 (1):\penalty0 33--55, 2000.

\bibitem[Yu(2006)]{Yuliang05}
L.~Yu.
\newblock Lowness for genericity.
\newblock \emph{Arch. Math. Logic}, 45:\penalty0 233–238, 2006.

\end{thebibliography}
\end{document}